\documentclass[runningheads]{llncs}

\usepackage[ngerman,english]{babel} %last input is standard
\usepackage[T1]{fontenc}
\usepackage[utf8]{inputenc}
\usepackage{lmodern}

\usepackage{amssymb}
\usepackage{amsmath}
\usepackage{amstext}
\usepackage{amsbsy}
\usepackage{amsfonts}
\usepackage{dsfont}
\usepackage{marvosym}

\usepackage[noadjust]{cite} %'noadjust': no automatic spaces before citations

\usepackage{enumitem}

\usepackage[colorlinks=true,linkcolor=blue,urlcolor=blue,citecolor=blue]{hyperref}

\usepackage{graphicx}

\usepackage{tabularx}
\usepackage{tabulary}
\usepackage{booktabs}
\usepackage{array}
\usepackage{multirow}
\usepackage{mathdots}
\usepackage{chngcntr}

\usepackage{tikz}
\usepackage{tkz-euclide}
\usetikzlibrary{positioning,calc,arrows,shapes,backgrounds}

\usepackage[ruled,linesnumbered]{algorithm2e}
\SetKwInOut{Input}{Input}
\SetKwInOut{Output}{Output}
\SetKwInOut{Set}{Set}

\newcommand{\wloss}{\text{w.l.o.g.\ }}
\newcommand{\etal}{\text{et al.\ }}
\newcommand{\eps}{\varepsilon}

\newcommand{\mbar}{\bar{m}}

\newcommand{\pjA}{p_j^A}

\newcommand{\plA}{p_l^A}
\newcommand{\pfA}{p_f^A}

\DeclareMathOperator*{\ALG}{ALG}
\DeclareMathOperator*{\OPT}{OPT}
\title{Scheduling with Testing on Multiple Identical Parallel Machines\thanks{Funded by the Deutsche Forschungsgemeinschaft (DFG, German Research Foundation) – 277991500/GRK2201, and by the European Research Council, Grant Agreement No. 691672.}}
\titlerunning{Scheduling with Testing on Identical Machines}
\author{Susanne Albers\inst{1}\and Alexander Eckl\inst{1,2,}\thanks{Corresponding author, eMail: alexander.eckl@tum.de}}
\authorrunning{S. Albers and A. Eckl}
\institute{Department of Informatics, Technical University of Munich,\\ Boltzmannstr. 3, 85748 Garching, Germany\\ \email{albers@in.tum.de, alexander.eckl@tum.de}\\
	\and
	Advanced Optimization in a Networked Economy, Technical University of Munich,\\ Arcisstra\ss e 21, 80333 Munich, Germany}

\begin{document}
\maketitle
\begin{abstract}
	Scheduling with testing is a recent online problem within the framework of explorable uncertainty motivated by environments where some preliminary action can influence the duration of a task. Jobs have an unknown processing time that can be explored by running a test. Alternatively, jobs can be executed for the duration of a given upper limit. We consider this problem within the setting of multiple identical parallel machines and present competitive deterministic algorithms and lower bounds for the objective of minimizing the makespan of the schedule. %The general setting is motivated for example by online user surveys for market prediction or querying centralized databases in distributed computing.
	In the non-preemptive setting, we present the SBS algorithm whose competitive ratio approaches $3.1016$ if the number of machines becomes large. We compare this result with a simple greedy strategy and a lower bound which approaches $2$. In the case of uniform testing times, we can improve the SBS algorithm to be $3$-competitive. For the preemptive case we provide a $2$-competitive algorithm and a tight lower bound which approaches the same value.
\end{abstract}

\section{Introduction}
\label{sec:introduction}

One of the most fundamental problems in online scheduling is makespan minimization on multiple parallel machines. An online sequence of $n$ jobs with processing times $p_j$ has to be assigned to $m$ identical machines. The objective is to minimize the makespan of the schedule, i.e.\ the maximum load on any machine. In 1966, Graham~\cite{Graham1966} showed that the List Scheduling algorithm, which assigns every job to the currently least loaded machine, is $(2-\frac{1}{m})$-competitive. Since then the upper bound has been improved multiple times, most recently to $1.9201$ by Fleischer and Wahl \cite{FleischerWahl2000}. At the same time, the lower bound has also been the focus of a lot of research, the current best result is $1.88$ by Rudin~\cite{Rudin2001}.

We consider this classical problem in the framework of \emph{explorable uncertainty}, where part the input is initially unknown to the algorithm and can be explored by investing resources which are added as costs to the objective function. Let $n$ jobs be given. Every job $j$ has a processing time $p_j$ and an upper bound $u_j$. It holds $0 \le p_j \le u_j$ for all $j$. Each job also has a testing time $t_j \ge 0$. A job can be executed on one of $m$ identical machines in one of two modes: It can either be run untested, which takes time $u_j$, or be tested and then executed, which takes a total time of $t_j + p_j$. The number of jobs $n$, as well as all testing times $t_j$ and upper bounds $u_j$ are known to the algorithm in the beginning. In particular, an algorithm can sort/order the jobs in a convenient way based on these parameters. The processing time $p_j$ for job $j$ is revealed once the test $t_j$ is completed. This \emph{scheduling with testing} setting has been recently studied by Dürr \etal \cite{DuerrEtAl2018}, and Albers and Eckl \cite{AlbersEckl2020} on a single machine.

We differentiate between \emph{preemptive} and \emph{non-preemptive} settings: If preemption is allowed, a job may be interrupted at any time, and then continued later on a possibly different machine. No two machines may work on the same job at the same time. In case a job is tested, any section of the test must be scheduled earlier than any section of the actual job processing. In the non-preemptive setting, a job assigned to a machine has be fully scheduled without interruption on this machine, independent of whether it is tested or not. We also introduce the notion of \emph{test-preemptive} scheduling, where a job can only be interrupted right after its test is completed.

Scheduling with testing is well-motivated by real world settings where a preliminary evaluation or operation can be executed to improve the duration or difficulty of a task. Examples for the case of multiple machines include e.g.\ a manufacturing plan where a number of jobs with uncertain length have to be assigned to multiple workers, or a distributed computing setting where tasks with unknown parameters have to be allocated to remote computing nodes by a central scheduler. Several examples for applicable settings for scheduling with testing can also be found in \cite{AlbersEckl2020,DuerrEtAl2018}.

In summary, we study the classical problem of makespan minimization on identical parallel machines in the framework of explorable uncertainty. We use competitive analysis to compare the value of an algorithm with an optimal off\-line solution. The setting closely relates to online machine scheduling problems studied previously in the literature. We investigate deterministic algorithms and lower bounds for the preemptive and non-preemptive variations of this problem.

\subsection{Related Work}

Scheduling with testing describes the setting where jobs with uncertain processing times have to be scheduled tested or untested on a given number of machines. The problem has been first studied by Dürr \etal \cite{DuerrEtAl2018,DuerrEtAl2020} for the special case of scheduling jobs on a single machine with uniform testing times $t_j \equiv 1$.
%Full-version exclusively
For the objective of minimizing the sum of completion times, they give a lower bound of 1.8546 and an upper bound of 2 in the deterministic setting. In the randomized setting, they present a lower bound of 1.6257 and a 1.7453-competitive algorithm. They also provide several upper bounds closer to the best possible ratio of 1.8546 for special case instances. Tight algorithms for the objective of minimizing the makespan are given for both the deterministic and randomized setting. 
More recently, Albers and Eckl \cite{AlbersEckl2020} considered the one machine case with testing times $t_j \ge 0$, presenting generalized algorithms for both objectives. In this paper, we consider scheduling with testing on identical parallel machines, a natural generalization of the previously studied one machine case.

\emph{Makespan minimization} in online scheduling with identical machines has been studied extensively in the past decades, ever since Graham \cite{Graham1966} established his $(2-\frac{1}{m})$-competitive List Scheduling algorithm in 1966. In the deterministic setting, a series of publications improved Graham's result to competitive ratios of $2-\frac{1}{m}-\eps_m$ \cite{GalambosWoeginger1993} where $\eps_m \to 0$ for large $m$, $1.985$ \cite{BartalEtAl1992}, $1.945$ \cite{KargerEtAl1996}, and $1.923$~\cite{Albers1999}, before Fleischer and Wahl \cite{FleischerWahl2000} presented the current best result of $1.9201$. In terms of the deterministic lower bound for general $m$, research has been just as fruitful. The bound was improved from $1.707$ \cite{FaigleEtAl1989}, to $1.837$ \cite{BartalEtAl1994}, and $1.852$ \cite{Albers1999}. The best currently known bound of $1.88$ is due to Rudin \cite{Rudin2001}. For the randomized variant, the lower bound has a current value of $\frac{e}{e-1} \approx 1.582$ \cite{ChenEtAl1994,Sgall1997}, while the upper bound is $1.916$ \cite{Albers2002}. For the deterministic preemptive setting, Chen \etal \cite{ChenEtAl1994b} provide a tight bound of $\frac{e}{e-1}$ for large values of $m$.

More recently, various extension of this basic case have emerged. In \emph{re\-source aug\-men\-ta\-tion} settings the algorithm receives some extra resources like machines with higher speed \cite{KalyanasundaramPruhs2000}, parallel schedules \cite{KellererEtAl1997,AlbersHellwig2017}, or a reordering buffer \cite{KellererEtAl1997, EnglertEtAl2008}. 
%Full-version only
In a related setting, the algorithm might be allowed to migrate jobs \cite{SandersEtAl2009}.
A variation that is closely related to our setting is \emph{semi-online scheduling}, where some additional piece of information is available to the online algorithm in advance. Possible pieces of information include for example the sum of all processing times \cite{KellererEtAl1997,AlbersHellwig2012,KellererEtAl2015}, the value of the optimum \cite{AzarRegev2001}, or information about the job order \cite{Graham1969}. Refer also to the survey by Epstein \cite{Epstein2018} for an overview of makespan minimization in semi-online scheduling.

Scheduling with testing is directly related to \emph{explorable uncertainty}, a research area concerned with obtaining additional information of unknown parameters through queries with a given cost. Kahan \cite{Kahan1991} pioneered this line of research in 1991 by studying approximation guarantees for the number of queries necessary to obtain the maximum and median value of a set of uncertain elements. Following this, a variety of problems have been studied in this setting, for example finding the median or $k$-smallest value \cite{FederEtAl2003,Khanna2001,GuptaEtAl2011}, geometric tasks \cite{BruceEtAl2005}, caching \cite{Olston2000}, as well as combinatorial problems like minimum spanning tree \cite{ErlebachEtAl2008,Megow2017}, shortest path \cite{FederEtAl2007}, and knapsack \cite{GoerigkEtAl2015}. We refer to the survey by Erlebach and Hoffmann \cite{ErlebachHoffmann2015} for an overview. In the scheduling with testing model, the cost of the queries is added to the objective function. Similar settings are considered for example in Weitzman's pandora's box problem \cite{Weitzman1979}, or in the recent 'price of information' model by Singla~\cite{Singla2018}.

\subsection{Contribution}

In this paper we provide the first results for makespan minimization on multiple machines with testing. We differentiate between general tests $t_j \ge 0$ and uniform tests $t_j = 1$, and consider non-preemptive as well as preemptive environments. In Table \ref{tab:results}, we illustrate our results for these cases. The parameter $m$ corresponds to the number of machines in the instance.
\begin{table}[htb]
	\centering
	\caption{Overview of results}
	\label{tab:results}
	\renewcommand{\arraystretch}{1.2}
	\begin{tabulary}{\textwidth}{C | C | C | C}
		%		\toprule
		\textbf{Setting} 	& \textbf{General tests}								& \textbf{Uniform tests}						& \textbf{Lower bound}							\\ \midrule
		Non-preemptive		& $c(m) \xrightarrow[m \to \infty]{} 3.1016$			& $c_1(m) \xrightarrow[m \to \infty]{} 3$		& $\max(\varphi, 2-\frac{1}{m})$				\\
		Preemptive			& $2$													& $2$											& $\max(\varphi, 2-\frac{2}{m}+\frac{1}{m^2})$	\\
		%		\bottomrule
	\end{tabulary}
\end{table}

In the non-preemptive setting, we present our main algorithm with competitive ratio $c(m)$, which we refer to as the \emph{SBS algorithm}. The function $c(m)$ is increasing in $m$ and has a value of approximately $3.1016$ for $m \to \infty$. For uniform tests, we can improve the algorithm to a competitive ratio of $c_1(m)$, which approaches $3$ for large values of $m$. Additionally, we analyze a simple Greedy algorithm for general tests with a competitive ratio of $\varphi (2-\frac{1}{m})$, where $\varphi \approx 1.6180$ is the golden ratio. We also provide a lower bound with value $\max(\varphi, 2-\frac{1}{m})$. The values of $c(m)$, $c_1(m)$, the Greedy algorithm and the lower bound are summarized in Table \ref{tab:non-p_results}. For all values of $m>1$ the SBS algorithm has better ratios compared to Greedy. At the same time, the uniform version of the algorithm improves these results further. Though our algorithms work for any number of machines $m$, they all achieve the same ratio for $m=1$ as was already proven in \cite{DuerrEtAl2018} and \cite{AlbersEckl2020} for uniform and general tests, respectively.

If the scheduler is allowed to use preemption, we obtain a $2$-approximation for both general and uniform tests. The result holds even in the more restrictive test-preemptive setting. The corresponding lower bound of $\max(\varphi, 2-\frac{2}{m}+\frac{1}{m^2})$ is tight when the number of machines becomes large.

\begin{table}[b]
	\centering
	\caption{Results in the non-preemptive setting for selected values of $m$}
	\label{tab:non-p_results}
	\begin{tabulary}{\textwidth}{C | C | C | C | C | C | C | C | C}
		%		\toprule
								& \textbf{1}& \textbf{2}& \textbf{3}& \textbf{4}& \textbf{5}& \textbf{10}	& \textbf{100}	&\textbf{$\infty$}	\\ \midrule
		Greedy					& $1.6180$	& $2.4271$	& $2.6967$	& $2.8316$	& $2.9125$	& $3.0743$		& $3.2199$		& $3.2361$			\\
		SBS						& $1.6180$	& $2.3806$	& $2.6235$	& $2.7439$	& $2.8158$	& $2.9591$		& $3.0874$		& $3.1016$			\\
		Uniform-SBS				& $1.6180$	& $2.3112$	& $2.5412$	& $2.6560$	& $2.7248$	& $2.8625$		& $2.9862$		& $3$				\\
		Lower Bound				& $1.6180$	& $1.6180$	& $1.6667$	& $1.75$	& $1.8$		& $1.9$			& $1.99$		& $2$
		%		\bottomrule
	\end{tabulary}
\end{table}

We utilize various methods for our algorithms and lower bounds. The Greedy algorithm we present is a variation of the well-known List Scheduling algorithm introduced by Graham \cite{Graham1966}. For the more involved SBS algorithm and its uniform version we employ testing rules for jobs based on the ratio between their upper bound and testing time similar to \cite{AlbersEckl2020}. We additionally divide the schedule into phases based on these ratios, therefore sorting the jobs by the given parameters to guarantee competitiveness. In the preemptive setting, we divide the schedule into two independent phases, testing and execution, and use an offline algorithm for makespan minimization to solve each instance separately. Lastly, the lower bounds we provide are loosely based on a common construction for the classical makespan minimization setting on multiple machines, where a large number of small jobs is followed by a single larger job.

The rest of the paper is structured in the following way: We start by giving some general definitions needed for later sections. In Section \ref{sec:non_preemptive} we then first prove the competitive ratio of Greedy and the lower bound, before describing the main algorithm for the general case. At the end of the section, we then build a special version of the algorithm for the uniform case. In Section \ref{sec:preemptive} we consider the preemptive setting and give an algorithm as well as a tight lower bound. We conclude the paper by describing some open problems.

\subsection{Preliminary Definitions}

We use the following notations throughout the document: For a job $j \in [n]$, the \emph{optimal offline running time} of $j$, i.e.\ the time needed by the optimum to schedule $j$ on a machine, is denoted as $\rho_j := \min(t_j + p_j,u_j)$, while the \emph{algorithmic running time} of $j$, i.e.\ the time needed for an algorithm to run $j$ on a machine, is given by
\begin{equation}
\label{eq:alg_proc_time}
\pjA := \begin{cases}
t_j + p_j 	& \text{if } j \text{ is tested,}\\
u_j 		& \text{if } j \text{ is not tested.}
\end{cases}
\end{equation}

It is clear that $\rho_j \le \pjA$ for any job $j$. Additionally, it holds that $p_j \le \rho_j$, since the processing times $p_j$ are upper bounded by $u_j$.

At times, we may use the definition of the \emph{minimal running time} of job~$j$, which is given by $\tau_j := \min(t_j,u_j)$.

It is clear that any job must fulfill $\tau_j \le \rho_j$. In total, we get the following estimation for the different running times:
\begin{equation}
\label{eq:running_time_estimation}
\tau_j \le \rho_j \le \pjA, \qquad \forall j \in [n]
\end{equation}

Since an algorithm does not know the values $p_j$, the testing decisions for the jobs are non-trivial. A partial goal for any competitive algorithm is to define a testing scheme such that the algorithmic running times are not too large compared to the optimal offline running times. We provide the following result which was used previously in \cite{AlbersEckl2020} and is based on Theorem 14 of \cite{DuerrEtAl2018}. The given testing scheme based on the ratio $r_j = u_j/t_j$ between upper bound and testing time is used multiple times within this paper.

\begin{proposition}
	\label{prop:claims_case_dist}
	Let job $j$ be tested iff $r_j \ge \alpha$ for some $\alpha \ge 1$. Then:
	\begin{enumerate}[label=\emph{(\alph*)}]
		\item $\forall j \in [n]$ tested: $\pjA \le \left(1+\frac{1}{\alpha}\right) \rho_j$
		\item $\forall j \in [n]$ not tested: $\pjA \le \alpha \rho_j$
	\end{enumerate}
\end{proposition}

As a direct consequence of Proposition \ref{prop:claims_case_dist}, an optimal testing scheme for a single job is given by setting the threshold $\alpha$ to the golden ratio $\varphi \approx 1.6180$ \cite{DuerrEtAl2018}.

\section{Non-preemptive Setting}
\label{sec:non_preemptive}

In this section we assume that preemption is not allowed. Any job has to be assigned to one of $m$ available machines. Since we only consider makespan minimization, we may assume that there is no idle time on the machines and the actual ordering of the executions on a machine does not influence the outcome of the objective. It is therefore sufficient to only consider the assignment of the jobs to the machines.

\subsection{Lower Bound and Greedy Algorithm}

We first prove a straightforward lower bound and extend the simple List Scheduling algorithm from the classical setting to our problem.

For the lower bound we choose negligibly small testing times coupled with very large upper bounds. This forces the algorithm to test all jobs and thus having to decide on a machine for a given job while having no information about its real execution time.

\begin{theorem}
\label{thm:simple_lower_bound}
	No online algorithm is better than $(2-\frac{1}{m})$-competitive for the problem of makespan minimization on $m$ identical machines with testing, even if all testing times are equal to $1$.
\end{theorem}

We note that $\varphi \approx 1.6180$ is always a lower bound for our problem (see \cite{DuerrEtAl2018}), which is relevant only for small values of $m \le 2$. The proof of Theorem \ref{thm:simple_lower_bound} is provided in Appendix \ref{app:sec:proofs}.

To prove a simple upper bound, we can generalize the List Scheduling algorithm to our problem variant as follows:

Consider the given jobs in any order. For a job $j$ to be scheduled next, test $j$ if and only if $u_j/t_j \ge \varphi$ and then execute it completely on the current least-loaded machine.

\begin{theorem}
\label{thm:greedy_algorithm}
	The extension of List Scheduling described above is $\varphi \, (2-\frac{1}{m})$-competitive for minimizing the makespan on $m$ identical machines with non-uniform testing, where $\varphi \approx 1.6180$ is the golden ratio. This analysis is tight.
\end{theorem}

The proof structure is similar to the proof of List Scheduling and uses common lower bounds for makespan minimization. We again refer to Appendix \ref{app:sec:proofs} for all details.

\subsection{SBS Algorithm}
\label{subsec:main_algorithm}

In this section we provide a $3.1016$-competitive algorithm for the non-preemptive setting. It assigns jobs into three classes $S_1, B,$ and $S_2$ based on their ratios between upper bounds and testing times.

Let $[n]$ be the set of all jobs. We define a threshold function $T(m)$ for all $m$ and divide the jobs into disjoint sets $[n] = B\, \dot\cup\, S$, where $S$ will be further subdivided into $S_1$ and $S_2$. The set $B$ corresponds to jobs where the ratio $r_j = u_j/t_j$ between upper bound and testing time is large, while jobs in $S$ have a small ratio. We define
\begin{equation*}
	\begin{aligned}
		B := &\left\{j \in [n]: r_j \ge T(m)\right\},\\
		S := &[n] \setminus B.\\
	\end{aligned}
\end{equation*}
For the set $S$, we would like the algorithm to be able to distinguish jobs based on their optimal offline running time $\rho_j$. Of course, without testing the algorithm does not know these values, so we instead use the minimal running time $\tau_j$, which can be computed directly using offline input only, to divide the set $S$ further.

We define $S_1 \subset S$, such that $|S_1| = \min(m,|S|)$ and $\forall j_1 \in S_1, j_2 \in S \setminus S_1$: $\tau_{j_1} \ge \tau_{j_2}$.  In other words, $S_1$ is the set of at most $m$ jobs in $S$ with the largest minimal running times. If this definition of $S_1$ is not unique, we may choose any such set. We set $S_2 := S \setminus S_1$. It follows that if $|S| \le m$, then $S_2 = \emptyset$.

The idea behind dividing $S$ into two sets is to identify the $m$ largest jobs according to minimal running time and schedule them first, each on a separate machine. This allows us to lower bound the runtime of the remaining jobs later in the schedule.

In Algorithm \ref{alg:SBS} we describe the SBS algorithm which solves the non-uniform case and works in three phases corresponding to the sets $S_1, B$ and $S_2$:

\begin{algorithm}[bht]
	$B \leftarrow \{j \in [n]: r_j \ge T(m)\}$\;	$S \leftarrow [n] \setminus B$\;
	$S_1 \leftarrow S' \subset S \text{ s.t.\ } |S'| = \min(m,|S|),\ \tau_{j_1} \ge \tau_{j_2}\ \forall j_1{\in}S', j_2{\in}S{\setminus}S'$\; $S_2 \leftarrow S \setminus S_1$\;
	\ForEach{$j\in S_1$}{
		\eIf{$r_j \ge \varphi$}
		{test and run $j$ on an empty machine\;}
		{run $j$ untested on an empty machine\;}
	}
	\ForEach{$j\in B$}{
		test and run $j$ on the current least-loaded machine\;
	}
%	sort jobs in $S_2$ by non-increasing $\tau_j$\;
	\ForEach{$j\in S_2$}{
		run $j$ untested on the current least-loaded machine\;
	}
	\caption{SBS algorithm}
	\label{alg:SBS}
\end{algorithm}

In order to have a non-trivial testing decision for jobs in $S_1$, it makes sense to require that $T(m) \ge \varphi$ for all $m$. More specifically, we will define the threshold function $T(m)$ in the non-uniform setting as follows:
\begin{equation*}
T(m) = \frac{(3 + \sqrt{5})m - 2 + \sqrt{(38 + 6\sqrt{5})m^2 - 4(11+\sqrt{5})m + 12}}{6m-2}
\end{equation*}

\begin{theorem}
\label{thm:main_alg}
	Let $T(m)$ be a parameter function of $m$ defined as above. The SBS algorithm is $T(m) \left(\frac{3}{2} - \frac{1}{2m} \right)$-competitive for minimizing the makespan on $m$ identical machines with non-uniform testing.
\end{theorem}

The function $T(m)$ is increasing for all $m \ge 1$ and fulfills $T(1) = \varphi$ as well as approximately $T(m) \rightarrow 2.0678$ for $m \to \infty$. The competitive ratio of the algorithm is explicitly given by
\begin{equation*}
c(m) = \frac{(3 + \sqrt{5})m - 2 + \sqrt{(38 + 6\sqrt{5})m^2 - 4(11+\sqrt{5})m + 12}}{4m}.
\end{equation*}
For this function we have $c(1) = \varphi$ as well as approximately $c(m) \rightarrow 3.1016$ if $m$ approaches infinity. Additionally, it holds that $c(m) < \varphi\left(2-\frac{1}{m}\right)$ for all $m > 1$.

\begin{proof}
	We assume \wloss that the job indices are sorted by non-increasing optimal offline running times $\rho_1 \ge \dots \ge \rho_n$. We denote the last job to finish in the schedule of the algorithm as $l$ and the minimum machine load before job $l$ as $t$. It follows that the value of the algorithm is $t + \plA$.
	
	The value of the optimum is at least as large as the average sum of the optimal offline running times, or
	\begin{equation}
	\label{eq:lower_bound_average}
	L := \frac{1}{m} \sum_{j \in [n]} \rho_j \le \OPT,
	\end{equation}
	since in any schedule at least one machine must have a load of at least this average. At the same time, we know that the optimum has to schedule every job on some machine:
	\begin{equation}
	\label{eq:lower_bound_max}
	\rho_j \le \OPT \quad \forall j \in [n]
	\end{equation}
	We also utilize another common lower bound in makespan minimization, which is the sum of the processing times of the $m$-th and $(m{+}1)$-th largest job. If there are at least $m+1$ jobs, then some machine has to schedule at least $2$ of these jobs:
	\begin{equation}
	\label{eq:lower_bound_m_m+1}
		\rho_m + \rho_{m+1} \le \OPT.
	\end{equation}
	Here, $\rho_j$ is defined as $0$ if the instance has less than $j$ jobs.
	
	We differentiate between jobs handled by the algorithm in different phases and bound the algorithmic running times against the optimal offline running times. We write $\pjA \le \alpha_j \rho_j$ and define different values for $\alpha_j$ depending on the set $j$ belongs to. It holds that
	\begin{equation}
	\label{eq:SBS_bounds}
		\alpha_j = 
		\begin{cases}
		\varphi 			&\text{if } j \in S_1,\\
		1 + \frac{1}{T(m)}	&\text{if } j \in B,\\
		T(m) 				&\text{if } j \in S_2,
		\end{cases}
	\end{equation}
	by Proposition~\ref{prop:claims_case_dist} and the testing strategy of the algorithm.
	
	The objective value of the algorithm depends on the set job $l$ belongs to, so we differentiate between three cases. The following proposition upper bounds the algorithmic value $\ALG = t + \plA$ for each of these cases:
	\begin{proposition}
	\label{prop:SBS_estimates}
		The value of the algorithm can be estimated as follows:
		\begin{equation*}
			\ALG \le
			\begin{cases}
				\varphi \OPT &\text{if } l \in S_1,\\
				\left( \varphi + \left(1 + \frac{1}{T(m)}\right) \left(1 - \frac{1}{m} \right) \right) \OPT &\text{if } l \in B,\\
				T(m) \left(\frac{3}{2} - \frac{1}{2m} \right) \OPT &\text{if } l \in S_2.
			\end{cases}
		\end{equation*}
	\end{proposition}
	
	To prove this proposition, we utilize the lower bounds \eqref{eq:lower_bound_average}-\eqref{eq:lower_bound_m_m+1} and the estimates \eqref{eq:SBS_bounds} for the value of $\alpha_j$. A critical step lies in the estimation of $\plA$ for $l \in S_2$, where we are able to lower bound $\tau_l$ using the size of the $m$-th and $(m{+}1)$-th largest job because the algorithm already ran $m$ jobs from $S_1$ in the beginning of the schedule. We refer to the appendix for a detailed proof.
	
	It remains to take the maximum over all three cases and minimize the value in dependence of $T(m)$. The value in the case $l \in S_1$ is always less than the values given by the other cases, therefore we only want to minimize
	\begin{equation*}
		\max \left(\varphi + \left(1 + \frac{1}{T(m)}\right) \left(1 - \frac{1}{m} \right),\ T(m) \left(\frac{3}{2} - \frac{1}{2m} \right) \right).
	\end{equation*}
	
	The left side of the maximum is decreasing in $T(m)$, while the right side is increasing. The minimal maximum is therefore attained when both sides are equal. It can be easily verified that for the given definition of the threshold function $T(m)$ both sides of the maximum are equal for all values of $m \ge 1$.
	
	It follows that the final ratio can be estimated by $\frac{\ALG}{\OPT} \le T(m) \left(\frac{3}{2} - \frac{1}{2m} \right)$. \qed
\end{proof}

\subsection{An Improved Algorithm for the Uniform Case}

The previous section established an algorithm with a competitive ratio of approximately $3.1016$. We now present an algorithm with a better ratio in the case when $t_j = 1$ for all jobs. We define the threshold function $T_1(m)$ as follows:
\begin{equation*}
T_1(m) = \frac{2m-1 + \sqrt{16m^2-14m+3}}{3m-1}
\end{equation*}

The \emph{Uniform-SBS} algorithm works as follows: Sort the jobs by non-increasing $u_j$. Go through the sorted list of jobs and put the next job on the machine with the lowest current load. A job $j$ is tested if $u_j \ge T_1(m)$, otherwise it is run untested.

\begin{theorem}
\label{thm:main_alg_uniform}
	Uniform-SBS is a $T_1(m)(\frac{3}{2}-\frac{1}{2m})$-competitive algorithm for uniform instances.
\end{theorem}

For uniform jobs with $t_j = 1$, sorting by non-increasing upper bound $u_j$ is consistent with sorting by non-increasing ratio $r_j$. Hence, Uniform-SBS is similar to the SBS algorithm reduced to the phases corresponding to the sets $B$ and $S$, where $S$ contains \emph{all} small jobs. The reason behind running the $m$ largest jobs of $S$ first in the SBS algorithm was to upper bound the remaining jobs in $S$. For uniform testing times, this bound can be achieved \emph{without} this special structure.

The function $T_1(m)$ is increasing for all $m \ge 1$ and fulfills $T_1(1) = \varphi$ as well as $T_1(m) \rightarrow 2$ for $m \to \infty$. Computing the competitive ratio explicitly yields
\begin{equation*}
	c_1(m) = \frac{2m-1 + \sqrt{16m^2-14m+3}}{2m}.
\end{equation*}
These values start from $c_1(1) = \varphi$ and approach $c_1(m) \rightarrow 3$ if $m \rightarrow \infty$. Additionally, it holds that $c_1(m) < c(m)$ for all $m > 1$. In other words, this special version of the algorithm is strictly better than the general SBS algorithm described in Section \ref{subsec:main_algorithm}. We defer the proof of Theorem \ref{thm:main_alg_uniform} to Appendix \ref{app:sec:proofs}.

\section{Results with Preemption}
\label{sec:preemptive}

In this section we assume that jobs can be preempted at any time during their execution. An interrupted job may be continued on a possibly different machine, but no two machines may work on the same job at the same time. Testing a job must be completely finished before any part of its execution can take place.

It makes sense to additionally consider the following stricter definition of preemption within scheduling with testing: Untested jobs must be run without interruption on a single machine. If a job is tested, its test must also be run without interruption on one machine. The execution after the test may then be run without interruption on a possibly different machine. We call this setting \emph{test-preemptive}, referring to the fact that the only place where we might preempt a job is exactly when its test is completed. From an application point of view, the test-preemptive setting is a natural extension of the non-preemptive setting, allowing the scheduler to reconsider the assignment of a job after receiving more information through the test.

Clearly, the difficulty of settings within scheduling with testing increases in the following order: preemptive, test-preemptive and non-preemptive. We now present the $2$-competitive \emph{Two Phases} algorithm for the test-preemptive setting, which can be applied directly to the ordinary preemptive case. Additionally, we construct a lower bound of $2-2/m+1/m^2$ for the ordinary preemptive case. This lower bound then also holds for test-preemption, and is therefore tight for both settings when the number of machines $m$ approaches infinity.

The Two Phases algorithm for the test-preemptive setting works as follows: Let OFF denote an optimal offline algorithm for makespan minimization on $m$ machines. In the first phase, the algorithm schedules all jobs for their minimal running time $\tau_j$ using the algorithm OFF. Herein, the algorithm tests all jobs except trivial jobs with $t_j > u_j$, where running the upper bound is optimal. In the second phase, all remaining jobs are already tested, hence the algorithm now knows all remaining processing times $p_j$. We then use the offline algorithm OFF again to schedule the remaining jobs optimally. Finally, the algorithm obliviously puts the second schedule on top of the first.

\begin{theorem}
\label{thm:alg_test_prmptv}
	The Two Phases algorithm is $2$-competitive for minimizing the makespan on $m$ machines with testing in the test-preemptive setting.
\end{theorem}

The proof makes use of the assumption that the algorithm has access to unlimited computational power, which is a common assumption in online optimization. If we do not give the online algorithm this power, the result is slightly worse, since offline makespan minimization is strongly NP-hard. We may then make use of the PTAS for offline makespan minimization by Hochbaum and Shmoys~\cite{HochbaumShmoys1987} to achieve a ratio of $2 + \eps$ for any $\eps > 0$, where the runtime of the algorithm increases exponentially with $1/\eps$. The complete version of the proof can be found in Appendix \ref{app:sec:proofs}.
%Full-version only
\begin{proof}[Proof sketch]
	Let OFF be any optimal offline algorithm for makespan minimization on $m$ machines. In the first phase, our algorithm tests all jobs except trivial jobs and schedules them for their minimal running time using OFF. The resulting value is bounded by the optimum of the original instance.
	
	In the second phase, we use the offline algorithm OFF again to schedule the remaining jobs optimally. The value of OFF is again bounded by the optimum.
	
	The algorithm obliviously puts the second schedule on top of the first. In the worst case the completion time of the entire schedule is the sum of the two sub-schedules. \qed
\end{proof}
For the lower bound result we now consider the standard preemptive setting where a job can be interrupted at any time.

\begin{theorem}
	In the preemptive setting, no online algorithm for makespan minimization on $m$ identical machines with testing can have a better competitive ratio than $2 - 2/m + 1/m^2$, even if all testing times are equal to $1$.
\end{theorem}

	We note that $\varphi \approx 1.6180$ also remains a lower bound even for the preemptive case, since two machines cannot run the same job concurrently. It holds $2-2/m+1/m^2 < \varphi$ only for values of $m \le 4$.

\begin{proof}
	Let us consider the following example: Let $M$ be a sufficiently large number and let $m(m-1)$ small jobs be given with $t_j =1, p_j = 0, u_j = M$ as well as one large job $f$ with $t_f = 1, p_f = m-1, u_f = M$. As argued in the proof of Theorem \ref{thm:simple_lower_bound}, OPT has a value of $m$ and we may assume that the algorithm tests every job.
	
	In the preemptive setting we required that any execution of the actual processing time of a job can only happen after its test is completed, therefore any job $j$ that finished testing at some time $t$ is completed not earlier than $t + p_j$. The adversary decides the processing time of $j$ by the following rule: If $t \ge m-1+1/m$ and job $f$ has not yet been assigned, set $p_j = m - 1$ (i.e.\ set $j = f$). Else, set $p_j = 0$.
	
	If the adversary assigns job $f$ at any point, then job $f$ finished testing at time $t \ge m-1+1/m$. It follows that
	\begin{equation*}
	\ALG \ge t + p_f \ge m-1+\frac{1}{m} + m-1 = 2m - 2 + \frac{1}{m}.
	\end{equation*}
	Hence the competitive ratio is at least $\frac{\ALG}{\OPT} \ge 2 - \frac{2}{m} + \frac{1}{m^2}$.
	
	All that remains is to show that this assignment of $f$ happens at some point during the runtime of the algorithm. Assume that this is not the case, i.e.\ all jobs finish testing earlier than $m-1+1/m$. The adversary sets all $p_j = 0$, hence it follows directly that all jobs are completely finished before $m-1+1/m$. But this means that the algorithmic solution has a value of $\ALG < m-1+1/m$.
	
	Since $t_j = 1$ for all jobs, we know that the average load $L$ fulfills
	\begin{equation*}
	L \ge \frac{1}{m} (m(m-1)+1) = m-1+1/m.
	\end{equation*}
	But $L$ is a lower bound on the optimal value of the instance, even in the preemptive setting, contradicting $\ALG < m-1+1/m$. \qed
\end{proof}

\section{Conclusion}
\label{sec:conclusion}

We presented algorithms and lower bounds for the problem of scheduling with testing on multiple identical parallel machines with the objective of minimizing the make\-span. Such settings arise whenever a preliminary action influences cost, duration or difficulty of a task. Our main results were a $3.1016$-com\-pet\-i\-tive algorithm for the non-preemptive case and a tight $2$-com\-pet\-i\-tive algorithm for the preemptive case if the number of machines becomes large.

Apart from closing the gaps between our ratios and the lower bounds, we propose the following consideration for future work: A natural generalization of our setting is to consider \emph{fully-online} arrivals, where jobs arrive one by one and have to be scheduled immediately. It is clear that this setting is at least as hard as the problem considered in this paper. In Appendix \ref{app:sec:fully_online}, we provide a simple lower bound with value $2$ for this generalization that holds for all values of $m \ge 2$. An upper bound is clearly given by the Greedy algorithm we provided in Section \ref{sec:non_preemptive}. Finding further algorithms or lower bounds for this new setting is a compelling direction for future research.

%BIBLIOGRAPHY
\bibliographystyle{splncs04}
\bibliography{Mult_Parallel_Machines_Full}

\clearpage
\appendix

\section{Results for the Fully-Online Setting}
\label{app:sec:fully_online}

As an additional consideration, we also give some results for the fully-online setting, where the jobs arrive sequentially one by one. Whenever a job $j$ arrives, its upper bound $u_j$ and testing time $t_j$ is revealed. Testing the job then reveals the processing time $p_j$.

In this section we provide improved lower bounds compared to the semi-online setting, for which the lower bound was given by $\max(\varphi,2-\frac{1}{m})$. Recall that this was tight for $m=1$. The following result gives a better bound for all instances with at least two machines.

\begin{theorem}
	Let $m \ge 2$. In the fully-online setting, no algorithm is better than $2$-competitive for the problem of makespan minimization on multiple identical machines with testing, even if all testing times are equal to $1$.
\end{theorem}

\begin{proof}
	We consider an instance with $m+1$ jobs where the first $m$ jobs have $t_j = u_j = 1$ for all $j$. Additionally, there is a single job $f$ with values $t_f=2, u_f=3$. The processing times of the first $m$ jobs are irrelevant, since it is clear that running a job untested is always optimal. If the algorithm tests the final job, then the adversary sets $p_f = 3$, otherwise it sets $p_f = 0$.
	
	The $m$ smaller jobs arrive first. If the algorithm stacks any two or more of these jobs on the same machines, then we have $\ALG / \OPT \ge 2$ for the partial instance consisting only of the first $m$ jobs.
	
	Hence assume this is not the case and the algorithm produces a flat schedule of height $1$ after the first $m$ jobs. If the final job is now part of the instance, the optimum puts two of the smaller jobs on the same machine and can run $f$ on its own machine.
	
	If the algorithm tests the final job then it has an algorithmic running time of $\pfA = t_f + p_f = 2+3$. The optimum runs the job untested, resulting in a final optimal value of $\rho_j = 3$. In total:
	\begin{equation*}
	\frac{\ALG}{\OPT} = \frac{1 + (2 + 3)}{3} = 2
	\end{equation*}
	On the other hand, if the algorithm runs $f$ untested, then the algorithmic running time is given by $\pfA = u_f = 3$, while the optimum tests the job yielding $\rho_f = 2$. This gives
	\begin{equation*}
	\frac{\ALG}{\OPT} = \frac{1 + 3}{2} = 2.
	\end{equation*} \qed
\end{proof}

We now want to improve this simple and direct lower bound to some value larger than $2$. It turns out this is increasingly harder if the number of machines $m$ increases. The reason for this difficulty lies in the typical construction of lower bound examples based on several 'rounds' of $m$ jobs where the algorithm is forced to produce flat schedules in order to be competitive. The example above also employs this construction in the first step.

We have not yet used the difficulty in deciding the testing strategy for such rounds of $m$ jobs. For $m=2$, we can improve the lower bound to a value close to $2.1$.

\begin{theorem}
	In the fully-online setting, no algorithm is better than $2.0953$-com\-pet\-i\-tive for the problem of makespan minimization on two identical machines with testing.
\end{theorem}

The proof uses parameter optimization based on the testing and running times of $m$ jobs in one 'round' and of the final job. Since the number of parameters in this construction increases exponentially in dependence of $m$, we were unable to extend this result to general values of $m$. In particular, it is not directly clear whether parameters for instances with larger values of $m$ can be chosen such that the same or a higher bound holds. We present the easiest case of $m=2$ as a stand-in for all results with small values of $m$ which are still computationally tractable.

\begin{proof}
	The counterexample consists of three jobs. The first job has a ratio between upper bound and testing time of $r_1 = u_1/t_1 = \varphi$. We may scale all remaining running times such that, independent of the testing decision of the algorithm for the first job, we can always assume that $p_1^A = \varphi$ and $\rho_1 = 1$.
	
	The running times of the second and third job are parameterized with the following values: ${t_2 = b, u_2 = c, t_3 = d, u_3 = e}$ with $0 \le b \le c$ and $0 \le d \le e$. The adversary always chooses $p_j$ such that the outcome is worst possible for the algorithm, that is $p_j = u_j$ if $j$ is tested and $p_j = 0$ otherwise.
	
	We start by considering the first 'round' of jobs, which consists only of the first and second job. We want the algorithm to schedule these jobs on two distinct machines. Hence we have to make sure that the competitive ratio is high in case the algorithm uses the same machine for both jobs. So assume for now that the algorithm does.
	
	Clearly, the optimum always uses both machines if there is only two jobs. Hence it has a value of $\max(1,\rho_2)$. If the algorithm tests the second job, then
	\begin{equation}
	\label{eq:lb_two_machines_1}
	\frac{\ALG}{\OPT} = \frac{\varphi + b+c}{\max(1,c)}.
	\end{equation}
	If the algorithm runs the second job untested, then
	\begin{equation}
	\label{eq:lb_two_machines_2}
	\frac{\ALG}{\OPT} = \frac{\varphi + c}{\max(1,b)}.
	\end{equation}
	
	These are the first two fractions we want to maximize. Assume now that the algorithm does \emph{not} use the same machine for the first two jobs. Then the third job arrives and will be scheduled on top of the smaller of the two previous jobs. This gives an algorithmic value of $\min(\varphi, p_2^A) + p_3^A$.
	
	We assume that the values of the parameters are such that the optimum puts jobs 1 and 2 on one machine and job 3 on the other. If this is actually not the case then the optimal value can only be smaller. Hence we have $\OPT \le \max(1+\rho_2,\rho_3)$.
	
	We now differentiate four cases corresponding to the testing decision of the algorithm with respect to jobs 2 and 3. The realizations of the processing times are chosen by the adversary as described above.
	
	Jobs 2 and 3 are tested. Then 
	\begin{equation}
	\label{eq:lb_two_machines_3}
	\frac{\ALG}{\OPT} \ge \frac{\min(\varphi,b+c) + d+e}{\max(1+c,e)}.
	\end{equation}
	
	Job 2 is tested and job 3 is not tested. Then 
	\begin{equation}
	\label{eq:lb_two_machines_4}
	\frac{\ALG}{\OPT} \ge \frac{\min(\varphi,b+c) + e}{\max(1+c,d)}.
	\end{equation}
	
	Job 2 is not tested and job 3 is tested. Then 
	\begin{equation}
	\label{eq:lb_two_machines_5}
	\frac{\ALG}{\OPT} \ge \frac{\min(\varphi,c) + d+e}{\max(1+b,e)}.
	\end{equation}
	
	Jobs 2 and 3 are not tested. Then 
	\begin{equation}
	\label{eq:lb_two_machines_6}
	\frac{\ALG}{\OPT} \ge \frac{\min(\varphi,c) + e}{\max(1+b,d)}.
	\end{equation}
	
	All that remains is optimizing the minimum value of \eqref{eq:lb_two_machines_1}-\eqref{eq:lb_two_machines_6}. We used numeric optimization and received values of $b=1, c=\varphi, d=\varphi^2$ and $e \approx 3.8675$. This yields a minimum of
	\begin{equation*}
	\min \left( \eqref{eq:lb_two_machines_1} - \eqref{eq:lb_two_machines_6} \right) \approx 2.0953.
	\end{equation*} \qed
\end{proof}

For $m=2$ it turned out that we may achieve a lower bound larger than $2$ with two equal-sized jobs in the first round. This changes as soon as $m \ge 3$, where three equal-sized jobs in the first round lead to a ratio of at most $2$ when the algorithm stacks two of these three jobs on the same machine. It is unclear whether this can be remedied for arbitrary values of $m$ by choosing suitable parameters.

\section{An Improved Result for Uniform Instances with a Small Number of Uncertain Jobs}
\label{app:sec:uniform_lambda_1}

For an additional result in the uniform setting we take a closer look at jobs whose lower bound is smaller than their testing time. We call any job with $u_j > 1$ (in the case of uniform testing) \emph{uncertain}. For these jobs, the algorithm has to make a non-trivial decision whether to test or not. For all other jobs, running the upper bound untested is optimal. Let $\lambda$ be ratio between uncertain jobs and machines, that is
\begin{equation*}
\lambda = \frac{|\{ j \in [n]: u_j > 1 \}|}{m}.
\end{equation*}

\begin{lemma}
	\label{lem:uniform_lambda_1}
	For instances with $\lambda \le 1$, there exists a $\varphi \, (\frac{4}{3} - \frac{1}{3m})$-competitive non-preemptive algorithm for the uniform testing case.
\end{lemma}

\begin{proof}
	Our algorithm for the uniform setting with $\lambda \le 1$ first tends to all uncertain jobs before considering any others. There are at most $m$ such uncertain jobs by the definition of $\lambda$. We make the algorithm simply assign them to one machine each. By Proposition~\ref{prop:claims_case_dist} we have $\pjA \le \varphi \rho_j$ for all uncertain jobs if we choose the parameter $\alpha$ equal to the golden ratio.
	
	We can now employ a simple trick to solve the rest of the instance: Since the only remaining jobs are those without uncertainty, the algorithm has complete information about the algorithmic running times of the instance, that is it knows all values $\pjA$, even those of jobs who are not yet scheduled. At this point we employ the \emph{Largest Processing Time} (LPT) algorithm. LPT is a $(\frac{4}{3}-\frac{1}{3m})$-approximation for makespan minimization on parallel machines \cite{Graham1969}.
	
	Our algorithm has only assigned at most one job per machine so far. Additionally, since $\pjA \ge \min(1+p_j,u_j) \ge 1$ for all uncertain jobs and $\pjA \le 1$ for all other jobs, we know that these already assigned jobs correspond to the largest jobs w.r.t.\ the algorithmic running times. Hence the algorithm can assign all remaining jobs such that the final assignment is exactly the same as in the solution given by the offline algorithm LPT.
	
	Let us denote for any algorithm $A$ the solution given by $A$ on the instance with running times $x_j$ as $A_{(x_j)}$. Then the final value of our online algorithm ALG fulfills
	\begin{equation*}
	\ALG = \text{LPT}_{\left(\pjA\right)} \le \left( \frac{4}{3} - \frac{1}{3m} \right) \text{OPT}_{\left(\pjA\right)} = \varphi \left( \frac{4}{3} - \frac{1}{3m} \right) \text{OPT}_{\left(\frac{1}{\varphi}\, \pjA \right)}.
	\end{equation*}
	
	As we have argued above, it holds that $\frac{1}{\varphi}\ \pjA \le \rho_j$ for all jobs. In particular, it follows that the instance with processing times $\frac{1}{\varphi}\ \pjA$ has an optimal solution which is not larger than the optimal solution of the instance with processing times $\rho_j$. Therefore,
	\begin{equation*}
	\varphi \left( \frac{4}{3} - \frac{1}{3m} \right) \text{OPT}_{\left(\frac{1}{\varphi}\, \pjA \right)} \le  \varphi \left( \frac{4}{3} - \frac{1}{3m} \right) \text{OPT}_{(\rho_j)}.
	\end{equation*}
	
	Altogether, it follows that $\ALG \le \varphi \left(\frac{4}{3} - \frac{1}{3m}\right) \OPT$. \qed
\end{proof}

\section{Proofs}
\label{app:sec:proofs}
\subsection{Proof of Theorem \ref{thm:simple_lower_bound}}

\begin{proof}
	Let $M$ be a sufficiently large number and consider the following instance: On $m$ machines we are given $m(m-1)$ small jobs with values $t_j = 1, p_j=0$ and $u_j = M$. Additionally, we are given a single large job $f$ with $t_f = 1, p_f = m-1, u_f = M$.
	
	The optimum tests all jobs and has a value of
	\begin{equation*}
	\OPT = m,
	\end{equation*}
	which is achieved by distributing all small jobs onto $m-1$ machines and running job $f$ on the final machine.
	
	It is immediately clear that if an algorithm decides to run any job untested, the ratio between the algorithmic solution and the optimum becomes larger as $M$ increases:
	\begin{equation*}
	\frac{\ALG}{\OPT} \ge \frac{M}{m} = \xrightarrow[M \to \infty]{} \infty
	\end{equation*}
	
	Hence assume that the algorithm tests everything. Since all jobs have the same testing times and upper bounds, the algorithm cannot distinguish between them. In particular, it does not know which one the large job $f$ is. Hence the adversary can decide the realization of the processing times whenever a job is being tested. Assume the algorithm runs some job $j$ on machine $\mbar$. Let $N(\mbar)$ be the current number of jobs on machine $\mbar$ excluding $j$. Then the adversary sets $p_j$ as follows:
	\begin{itemize}
		\item If $N(\mbar) \ge m-1$ and job $f$ is not yet run, set $p_j = m - 1$ (i.e.\ set $j = f$).
		\item Else, set $p_j = 0$.
	\end{itemize}
	
	If at any point the algorithm tests some job $j$ and the corresponding machine $\mbar$ fulfills $N(\mbar) \ge m-1$ for the first time, then the adversary sets $j = f$ and we have
	\begin{equation*}
	\ALG \ge N(\mbar) + m \ge (m-1) + m = 2 m - 1.
	\end{equation*}
	
	The competitive ratio of the algorithm is then given by
	\begin{equation*}
	\frac{\ALG}{\OPT} \ge \frac{2 m - 1}{m} = 2 - \frac{1}{m}.
	\end{equation*}
	
	It remains to show that at some point the number of jobs on all machines is at least $m-1$, and hence the algorithm is forced to run the next job on such a machine. Assume this is not the case and the adversary declares the processing times of all jobs to be small. The average load on the machines after $m(m-1)+1$ such jobs is given by $(m(m-1)+1)/m > m-1$. Since the adversary has not set any job as large, all machines $m_i$ must also have a load of at most $m-1$, which is a contradiction. \qed
\end{proof}

\subsection{Proof of Theorem \ref{thm:greedy_algorithm}}

\begin{proof}
	The value of the optimum is at least as large as the average sum of the optimal offline running times, or
	\begin{equation*}
	L := \frac{1}{m} \sum_{j \in [n]} \rho_j \le \OPT.
	\end{equation*}
	At the same time, we know that the optimum has to at least schedule every job on some machine:
	\begin{equation*}
	\rho_j \le \OPT \quad \forall j \in [n]
	\end{equation*}
	We set $\alpha = \varphi$ in Proposition \ref{prop:claims_case_dist} and combine parts (a) and (b) to bound the algorithmic running time:
	\begin{equation*}
	\pjA \le \varphi \rho_j \quad \forall j \in [n]
	\end{equation*}
	
	Let $l$ be the job that finishes last in the schedule. Let $t$ be the minimum machine load right before $l$ is assigned. It follows that job $l$ starts at time $t$ and finishes at time $t + \plA$. This implies that the value of the algorithm is equal to $t + \plA$ as well.
	
	The value of $t$ is at most the average sum of algorithmic running times of all jobs scheduled before $l$. We overestimate this average using all jobs except $l$ itself. We receive
	\begin{equation*}
	\begin{aligned}
	t 	&\le \frac{1}{m} \sum_{j \neq l} \pjA \\
	&\le \frac{\varphi}{m} \sum_{j \neq l} \rho_j \\
	&= \frac{\varphi}{m} \sum_{j \in [n]} \rho_j - \frac{\varphi}{m} \rho_l \\
	&= \varphi L - \frac{\varphi}{m} \rho_l
	\end{aligned}
	\end{equation*}
	
	We can estimate the algorithmic value by
	\begin{equation*}
	\begin{aligned}
	\ALG 	&= t + \plA \\
	&\le \varphi L - \frac{\varphi}{m} \rho_l + \varphi \rho_l \\
	&\le \varphi L + \varphi \left(1 - \frac{1}{m} \right) \rho_l \\
	&\le \varphi \left(2 - \frac{1}{m} \right) \OPT,
	\end{aligned}
	\end{equation*}
	where we used the lower bounds \eqref{eq:lower_bound_average} and \eqref{eq:lower_bound_max} in the last step.
	
	Finally, we provide a short example to see that the above analysis is tight. We note that the counterexample depends on the fact that the algorithm does not sort the jobs. It is unclear whether Greedy with some additional sorting strategy yields a provably better result. However, since Greedy without sorting only considers one job after the other, it is directly applicable to the fully online case (see Section \ref{app:sec:fully_online}).
	
	Consider $m(m-1)$ small jobs with $t_j = 1, p_j = u_j = \varphi$ and a single large job $f$ with $t_f = m, p_f = u_f = \varphi m$. It is clear that the optimal makespan is given by $\varphi m$.
	
	Since Greedy doesn't sort the jobs, we may assume it tests and schedules all small jobs first. Afterwards, all machines have a load of $(m-1) (1+\varphi)$. Then, job $f$ is tested and run, yielding a final makespan of
	\begin{equation*}
		(m-1) (1+\varphi) + m (1+\varphi) = (2m-1) \varphi^2.
	\end{equation*} \qed
\end{proof}

\subsection{Proof of Proposition \ref{prop:SBS_estimates}}

\begin{proof}
Let the final job $l$ and the minimum machine load $t$ before job $l$ be defined as in the proof of Theorem \ref{thm:main_alg}. We want to estimate the value of the algorithm $\ALG = t + \plA$.

The value of $t$ is bounded by the average of the algorithmic running times of all jobs before $l$. Let $J_l$ be the set of jobs the algorithm assigns before $l$. Then:
\begin{equation*}
	\begin{aligned}
	t 	&\le \frac{1}{m} \sum_{j \in J_l} \pjA \\
		&= \frac{1}{m} \sum_{j \in J_l \cup {l}} \pjA - \frac{\plA}{m} \\
		&\le \frac{1}{m} \sum_{j \in J_l \cup {l}} \alpha_j \rho_j - \frac{\plA}{m} \\
	\end{aligned}
\end{equation*}

\emph{Case 1}: Job $l$ is in $S_1$. In this case, $l$ is the first job assigned to its machine by the definition of the assignment for jobs in this set. Since $l$ is also the last job on its machine, it follows that $l$ is the \emph{only} job on its machine and hence $t = 0$. By \eqref{eq:SBS_bounds} and \eqref{eq:lower_bound_max} we have
\begin{equation*}
t + \plA \le \alpha_l \rho_l = \varphi \rho_l \le \varphi \OPT.
\end{equation*}

\emph{Case 2}: $l \in B$. In this case the set $J_l$ only contains jobs from $S_1$ and $B$. Since $l$ itself is also in $B$, we can use \eqref{eq:SBS_bounds} to write
\begin{equation*}
\begin{aligned}
t 	&\le \frac{1}{m} \sum_{j \in J_l \cup {l}} \alpha_j \rho_j - \frac{\plA}{m}\\
	&\le \frac{1}{m} \left( \sum_{j \in S_1} \varphi \rho_j + \sum_{j \in B} \left(1 + \frac{1}{T(m)}\right) \rho_j \right) - \frac{\plA}{m}\\
	&\le \frac{\varphi}{m} \sum_{j \in [n]} \rho_j - \frac{\plA}{m}\\
	&= \varphi L - \frac{\plA}{m},
\end{aligned}
\end{equation*}
where we additionally used $1 + 1/T(m) \le \varphi$.

For the value of the algorithm we use \eqref{eq:SBS_bounds} again to receive
\begin{equation*}
\begin{aligned}
t + \plA&\le \varphi L - \frac{\plA}{m} + \plA\\
		&\le \varphi L + \left(1 + \frac{1}{T(m)} \right) \left(1 - \frac{1}{m} \right) \rho_l\\
		&\le \left( \varphi + \left(1 + \frac{1}{T(m)}\right) \left(1 - \frac{1}{m} \right) \right) \OPT.\\
\end{aligned}
\end{equation*}
Here we additionally used \eqref{eq:lower_bound_average} and \eqref{eq:lower_bound_max} in the final step.

\emph{Case 3}: $l \in S_2$. The set $J_l$ may now contain jobs of any set. We estimate $t$ as best as possible using \eqref{eq:SBS_bounds}. Since $T(m) \ge \varphi \ge 1 + 1/T(m)$, we have
\begin{equation*}
\begin{aligned}
t 	&\le \frac{1}{m} \sum_{j \in J_l \cup {l}} \alpha_j \rho_j - \frac{\plA}{m}\\
%	&\le \frac{1}{m} \left( \sum_{j \in S_1} \varphi \rho_j + \sum_{j \in B} \left(1 + \frac{1}{T(m)}\right) \rho_j + \sum_{j \in S_2} T(m) \rho_j  \right) - \frac{\plA}{m} \\
	&\le \frac{T(m)}{m} \sum_{j \in [n]} \rho_j - \frac{\plA}{m} \\
	&= T(m) L - \frac{\plA}{m}.
\end{aligned}
\end{equation*}

To receive the desired competitive ratio, we want to estimate $\plA$. It now becomes apparent why we chose to schedule the $m$ largest jobs w.r.t.\ the minimal running time in the first phase of the algorithm: Since $l$ is in the set $S_2$ (and therefore $S_2$ is not empty), we know that $|S_1| = m$ and these $m$ jobs have a minimal running time not smaller than $l$.

Since the $\tau_j$ are lower bounds for the optimal offline running times $\rho_j$, it follows that $\tau_l \le \rho_j$ for all jobs $j \in S_1$. Including $l$ itself there are at least $m+1$ such jobs. In particular, using the sorting of the optimal offline running times, we have $\tau_l \le \rho_m$ and $\tau_l \le \rho_{m+1}$ for the $m$-th and $(m{+}1)$-th largest job. With equation \eqref{eq:lower_bound_m_m+1}, we receive
\begin{equation*}
\tau_l \le \frac{1}{2} \left(\rho_m + \rho_{m+1}\right) \le \frac{\OPT}{2}.
\end{equation*}

If $\tau_l = u_l$, then it follows directly that $\plA = u_l = \tau_l \le \frac{\OPT}{2}$. If on the other hand $\tau_l = t_l$, then, since $l$ is in $S$, we have $\plA = u_l < T(m) \cdot t_l = T(m) \cdot \tau_l \le \frac{T(m)}{2} \OPT$. Because of $T(m) \ge 1$ it follows in both cases that	
\begin{equation*}
\plA = u_l \le \frac{T(m)}{2} \OPT.
\end{equation*}

The value of the algorithm is then
\begin{equation*}
\begin{aligned}
t + \plA&\le T(m) L - \frac{\plA}{m} + \plA\\
		&\le T(m) \OPT + \left(1 - \frac{1}{m} \right) \frac{T(m)}{2} \OPT\\
		&\le T(m) \left(\frac{3}{2} - \frac{1}{2m} \right) \OPT.
\end{aligned}
\end{equation*}
This concludes the proof of the proposition. \qed
\end{proof}

\subsection{Proof of Theorem \ref{thm:main_alg_uniform}}

\begin{proof}
	As before, let $L$ be the lower bound \eqref{eq:lower_bound_average}. The lower bounds \eqref{eq:lower_bound_max} and \eqref{eq:lower_bound_m_m+1} also hold. We again denote the last job to finish as $l$ and the minimum machine load before $l$ as $t$. Hence, the value of the algorithm is $t + \plA$.
	
	By Proposition \ref{prop:claims_case_dist}, the testing scheme of the algorithm yields $\pjA \le \alpha_j \rho_j$, where
	\begin{equation*}
		\alpha_j := 
		\begin{cases}
			1 + \frac{1}{T_1(m)} 	&\text{if } j \text{ is tested,} \\
			T_1(m)				&\text{else.}
		\end{cases}
	\end{equation*}
	
	We first deal with the case when the number of jobs $n$ is less than or equal to the number of machines $m$. In this case, the algorithm puts at most one job on every machine. Consider job $l$, the last job to finish. By the testing scheme of the algorithm it holds that
	\begin{equation*}
		\ALG = \plA \le \max \left( 1 + \frac{1}{T_1(m)}, T_1(m) \right) \rho_j \le T_1(m) \rho_j \le T_1(m) \OPT,
	\end{equation*}
	where the second inequality holds due to $T_1(m) \ge \varphi$ and the last due to equation \eqref{eq:lower_bound_max}. This concludes the special case where $n \le m$.
	
	Let us now consider $n > m$. We assume \wloss that the job indices are sorted by non-increasing optimal offline running times $\rho_j$, i.e.\ $\rho_1 \ge \dots \ge \rho_n$. Since we now have at least $m+1$ jobs, the lower bound of the $m$-th and $(m{+}1)$-th largest job is applicable.
	
	We bound the value of $t$ by the average of the algorithmic running times of all jobs run before $l$. Let $J_l$ be the set of jobs the algorithm assigns before $l$.
	\begin{equation*}
		\begin{aligned}
			t 	&\le \frac{1}{m} \sum_{j \in J_l} \pjA \\
			&= \frac{1}{m} \sum_{j \in J_l \cup {l}} \pjA - \frac{\plA}{m} \\
			&\le \frac{1}{m} \sum_{j \in [n]} \alpha_j \rho_j - \frac{\plA}{m} \\
		\end{aligned}
	\end{equation*}
	
	\emph{Case 1}: The algorithm tests job $l$. Then, by the non-increasing order of the upper bounds, all jobs in $J_l$ are tested as well. Hence $\alpha_j = 1+1/T_1(m)$ for all $j \in J_l$. Combining this with \eqref{eq:lower_bound_average}, we get	
	\begin{equation*}
		\begin{aligned}
			t 	&\le \frac{1}{m} \sum_{j \in [n]} \left(1 + \frac{1}{T_1(m)}\right) \rho_j - \frac{\plA}{m} \\
				&= \left(1 + \frac{1}{T_1(m)}\right) L - \frac{\plA}{m}.
		\end{aligned}
	\end{equation*}
	
	Finally, since $l$ itself is also tested, we can write
	\begin{equation*}
		\begin{aligned}
			\ALG 	&= t + \plA \\
			&\le \left(1 + \frac{1}{T_1(m)}\right) L - \frac{\plA}{m} + \plA \\
			&\le \left(1 + \frac{1}{T_1(m)}\right) \left(L + \left(1 - \frac{1}{m} \right) \rho_l \right) \\
			&\le \left(1 + \frac{1}{T_1(m)}\right) \left(2 - \frac{1}{m} \right) \OPT. \\
		\end{aligned}
	\end{equation*}
	
	\emph{Case 2}: The algorithm runs $l$ untested. If $l$ is part of the first round of $m$ jobs, that is if $l$ is the only job on its machine, then we can argue analogously to the case $n \le m$ that $\ALG = \plA \le T_1(m) \OPT$.
	
	Otherwise, recall the definition of the minimal running time $\tau_j$ of job~$j$, which is
	\begin{equation*}
		\tau_j = \min(1,u_j)
	\end{equation*}
	in the uniform testing case. As we argued previously, it holds $\rho_j \ge \tau_j$ for all jobs.
		
	Since there are at least $m$ jobs the algorithm considers before $l$, and the algorithm sorts all jobs by $u_j$, we know that there exist at least $m+1$ jobs $j$ with $\tau_j \ge \tau_l$, including $l$ itself. Since the $\tau_j$ are lower bounds for the optimal offline running times $\rho_j$, it follows that $\tau_j \le \rho_l$ for at least $m+1$ different jobs $j$. Using the sorting of the optimal offline running times and \eqref{eq:lower_bound_m_m+1}, we receive
	\begin{equation*}
		\tau_l \le \frac{1}{2} \left(\rho_m + \rho_{m+1}\right) \le \frac{\OPT}{2}.
	\end{equation*}
	
	Since $l$ is not tested by the algorithm we have $\plA = u_l < T_1(m)$. Now, if $\tau_l = u_l$, then it follows directly that $\plA = u_l = \tau_l \le \frac{\OPT}{2}$. If on the other hand $\tau_l = 1$, then $\plA < T_1(m) = T_1(m) \cdot \tau_l \le \frac{T_1(m)}{2} \OPT$. Since $T_1(m) \ge 1$ it follows in both cases that	
	\begin{equation*}
		\plA = u_l \le \frac{T_1(m)}{2} \OPT.
	\end{equation*}
	
	We do not know which jobs before $l$ the algorithm tests or runs untested. From $T_1(m) \ge \varphi$ follows $\alpha_j \le T_1(m)$ in both cases, hence we have $\pjA \le \alpha_j \rho_j \le T_1(m) \rho_j$ for all $j \in [n]$. We write	
	\begin{equation*}
		\begin{aligned}
			t 	&\le \frac{1}{m} \sum_{j \in [n]} T_1(m) \rho_j - \frac{\plA}{m} \\
				&= T_1(m) L - \frac{\plA}{m}.
		\end{aligned}
	\end{equation*}
	
	For the algorithmic value:
	\begin{equation*}
		\begin{aligned}
			\ALG 	&= t + \plA \\
			&\le T_1(m) L - \frac{\plA}{m} + \plA \\
			&\le T_1(m) \OPT + \left(1 - \frac{1}{m} \right) \frac{T_1(m)}{2} \OPT \\
			&\le T_1(m) \left(\frac{3}{2} - \frac{1}{2m} \right) \OPT
		\end{aligned}
	\end{equation*}
	
	Finally, we take the maximum over the two cases above and minimize the value in dependence of $T_1(m)$. In other words we want to minimize
	\begin{equation*}
		\max \left( \left(1 + \frac{1}{T_1(m)}\right) \left(2 - \frac{1}{m} \right),\ T_1(m) \left(\frac{3}{2} - \frac{1}{2m} \right) \right).
	\end{equation*}
	
	The left side of the maximum is decreasing in $T_1(m)$, while the right side is increasing. The minimal maximum is attained when both sides are the same. The definition of $T_1(m)$ balances both sides, which can be easily verified by inserting.
	
	It follows that the ratio of the algorithm can be estimated by
	\begin{equation*}
		\frac{\ALG}{\OPT} \le T_1(m) \left(\frac{3}{2} - \frac{1}{2m} \right).
	\end{equation*} \qed
\end{proof}

\subsection{Proof of Theorem \ref{thm:alg_test_prmptv}}

\begin{proof}
	We assume for now that the algorithm has access to unlimited computational power. Let OFF be any optimal offline algorithm for the makespan minimization problem on $m$ identical machines. As the name suggests, the Two Phases algorithm divides the instance into two phases, which are scheduled one after the other and are both bounded by the value of the optimal solution, thus giving a $2$-competitive result.
	
	In the first phase, our algorithm tests all jobs except trivial jobs with $t_j > u_j$, where running the upper bound is always optimal. Recall that the minimal running time $\tau_j$ of a job $j$ was defined as $\min(t_j,u_j)$. The algorithm then schedules all jobs for their minimal running time using OFF. 
	
	Since the algorithm has access to unlimited computational power, it is clear that OFF is able to return the optimal solution on the instance with input $\tau_j$. At the same time, as we have seen in equation \eqref{eq:running_time_estimation}, any job fulfills $\tau_j \le \rho_j$, hence this optimal value can only be smaller than or equal to the optimum of the original instance. We again denote for any algorithm $A$ the solution given by $A$ on the instance with running times $x_j$ as $A_{(x_j)}$. Then,
	\begin{equation*}
	\text{OFF}_{(\tau_j)} = \text{OPT}_{(\tau_j)} \le \OPT.
	\end{equation*}
	
	All trivial jobs with $t_j > u_j$ are now already completely scheduled. Since all remaining jobs were tested in the first phase, the algorithm now knows all remaining processing times $p_j$. We can therefore use the offline algorithm OFF again to schedule the remaining jobs optimally. To make notation easier, assume $p_j = 0$ for all jobs which are already finished. Hence,
	\begin{equation*}
	\text{OFF}_{(p_j)} = \text{OPT}_{(p_j)} \le \OPT,
	\end{equation*}
	where the final step follows from the fact that $p_j \le \rho_j$ for all jobs.
	
	The algorithm obliviously puts the second schedule on top of the first, such that no job test overlaps its execution. In the worst case the completion time of the entire schedule is the sum of the two sub-schedules, or
	\begin{equation*}
	\ALG \le \text{OFF}_{(\tau_j)} + \text{OFF}_{(p_j)} \le 2 \OPT.
	\end{equation*}
	
	Finally, we consider the case where the algorithm must run in polynomial time. Then we can use the polynomial-time approximation scheme from \cite{HochbaumShmoys1987} in place of our offline algorithm OFF. Let $\eps>0$ be any small positive number. We run the PTAS by \cite{HochbaumShmoys1987} for both phases with a value of $\eps/2$ to receive
	\begin{equation*}
	\ALG \le \text{PTAS}_{(\tau_j)} + \text{PTAS}_{(p_j)} \le (1 + \eps/2) \OPT + (1 + \eps/2) \OPT \le (2 + \eps) \OPT.
	\end{equation*}
	Hence we have an algorithm that runs in polynomial time in the input and is $(2+\eps)$-competitive for any fixed $\eps>0$. \qed
\end{proof}

\end{document}